\documentclass[runningheads]{llncs}
\usepackage[T1]{fontenc}
\usepackage{amsmath}
\usepackage{multicol}
\usepackage{amsfonts}
\usepackage{xcolor}
\usepackage{graphicx}
\usepackage{tcolorbox}

\usepackage{tcolorbox}
\usepackage{tikz}
\usepackage[linesnumbered, ruled, vlined]{algorithm2e}
\usepackage{tikz}
\usepackage{caption}
\usepackage{subcaption}
\usepackage{cite}

\usepackage{url}
\usepackage{hyperref}
\hypersetup{
    colorlinks=true,
    linkcolor=blue,
    filecolor=magenta,      
    urlcolor=cyan,
}
\begin{document}
\title{An Effective Tag Assignment Approach for Billboard Advertisement}
\author{Dildar Ali\inst{1} \and Harishchandra Kumar\inst{2} \and Suman Banerjee\inst{1} \and Yamuna Prasad\inst{1}}
\authorrunning{Ali et al.} 
\institute{Indian Institute of Technology Jammu,
J \& K-181221, India. \and National Institute of Technology Raipur, Chhattisgarh, India\\
\email{2021rcs2009@iitjammu.ac.in, hkumar049.btech2021@cse.nitrr.ac.in, suman.banerjee@iitjammu.ac.in, yamuna.prasad@iitjammu.ac.in}}
\maketitle
\begin{abstract}
  Billboard Advertisement has gained popularity due to its significant outrage in return on investment. To make this advertisement approach more effective, the relevant information about the product needs to be reached to the relevant set of people. This can be achieved if the relevant set of tags can be mapped to the correct slots. Formally, we call this problem the \textsc{Tag Assignment Problem in Billboard Advertisement}. Given trajectory, billboard database, and a set of selected billboard slots and tags, this problem asks to output a mapping of selected tags to the selected slots so that the influence is maximized. We model this as a variant of traditional bipartite matching called \textsc{One-To-Many Bipartite Matching (OMBM)}. Unlike traditional bipartite matching, a tag can be assigned to only one slot; in the \textsc{OMBM}, a tag can be assigned to multiple slots while the vice versa can not happen. We propose an iterative solution approach that incrementally allocates the tags to the slots. The proposed methodology has been explained with an illustrated example. A complexity analysis of the proposed solution approach has also been conducted. The experimental results on real-world trajectory and billboard datasets prove our claim on the effectiveness and efficiency of the proposed solution. 
\end{abstract}
\section{Introduction}
In the past few years, advertisement has become a central goal for any E-commerce house as it gives a way to convince people about their products or different social and political events, and because of this, their products and ticket sales increase, i.e., revenue increases. The existing literature reported that around $7 - 10\%$ of total revenue is utilized by an E-commerce house for advertising purposes\footnote{\url{https://www.lamar.com/howtoadvertise/Research/}}. Among the several advertising techniques (e.g., newspapers, social networks, television, billboards, etc), billboard advertising has emerged as an effective advertising technique as it is easy to place advertisements and ensures more investment returns\footnote{\url{https://www.thebusinessresearchcompany.com/report/billboard-and-outdoor-advertising-global-market-report}}. Billboard Advertisement is a technique for creating a promotional advertisement for products, events, etc., with the hope that if the content is displayed to a relevant group of people, a subset of them will be influenced, and they might buy the product or attend the event. This will help increase revenue or promote the event. The advertisement content is formalized as tags, and different tags are relevant to different sets of people. So, the tag plays an important role in determining the influence of a billboard. In real-life scenarios, billboard advertisements require the selection of slots and tags, and although the influence value depends on billboard slots and tags, the literature is limited. However, a few studies have considered using tags to maximize their influence on social networks \cite{tekawade2023influence,banerjee2022budgeted}. Ke at al.\cite{ke2018finding} first studied the problem of finding $\ell$ many tags and $k$ many users in a social network to maximize the influence. In the context of the billboard advertisement, Ali et al. \cite{ali2024influential} first introduce the problem of finding the most influential billboard slots and tags. Knowing which tag is important to which slot to maximize the influence is essential in practical scenarios. This is the key problem addressed in this paper. 

\paragraph{\textbf{Motivation.}}
To make an advertisement effective, it is important to adopt a suitable approach such that the information reaches the group of relevant people. In real life, different tags are relevant to various categories of people. Consider a political party running a campaign for an upcoming election and proposing different projects in their election manifesto for different categories of people. For example, this might include doubling the farmer's income in the next three years and a certain percentage tax benefit to the industrialists. Now, the tags relevant to the benefits of the farmers must be displayed on the billboard slots where the audience is farmers. So, to address the real-life scenario study of \textsc{Tag Assignment Problem in Billboard Advertisement} is essential.

\paragraph{\textbf{Our Contribution.}} 
In billboard advertising, the influence varies greatly depending on which tags are displayed in which slots. Studies like \cite{ali2022influential,ali2023influential} have explored finding the top-$k$ influential billboard slots using various techniques. Recently, Ali et al. \cite{ali2024influential} proposed methods to identify influential slots and tags jointly.  However, no existing literature specifically addresses the tag assignment problem in billboard advertisements. This paper tackles this issue and proposes an iterative approach for incrementally allocating tags to slots.
In particular, we make the following contributions in this paper:

\begin{itemize}
    \item We model the Tag Assignment Problem as a two-sided matching problem and show that this problem is NP-hard.
    \item We introduce One-to-Many Bipartite Matching and propose an iterative solution for the tag assignment problem.
    \item We analyzed the time and space requirement of the proposed solution and provided an approximation guarantee.
    \item  The proposed approaches have been implemented with real-world trajectory and billboard datasets to highlight their effectiveness and efficiency.
\end{itemize}

\par Rest of the paper has been organized as follows. Section \ref{Problem} describes the required preliminary concepts and defines our problem formally. Section \ref{Sec:PSA} contains the proposed solution approaches with a detailed analysis. The experimental evaluation has been described in Section \ref{Sec:Experimental_Evaluation}. Finally, Section \ref{Sec:Conclusion} concludes our study and gives directions for future research.

\section{Preliminaries and Problem Definition} \label{Problem}
This section describes the preliminary concepts and formally defines our problem. Initially, we start by explaining the billboard advertisement and its procedure.

\subsection{Billboard Advertisement} A trajectory database contains the location information of trajectories of a particular city, and it can be stated in Definition \ref{Def:Trajectory}.

\begin{definition}[Trajectory Database] \label{Def:Trajectory}
The trajectory database $\mathcal{D}$ is a collection of a tuple in the form of $<\mathcal{U}_{id}, \mathcal{U}_{loc}, [t_{1},t_{2}]>$, where $\mathcal{U}_{id}$, $\mathcal{U}_{loc}$ and $[t_{1},t_{2}]$ represents the unique user id, user location and time stamp, respectively.
\end{definition}

For example, there is a tuple $<\mathcal{U}_{9}, Kolkata\_Airport, [200, 300]>$ and it denotes user $\mathcal{U}_{9}$ was in the $Kolkata\_Airport$ for the duration of $[200, 300]$. Next, we define the billboard database in Definition \ref{Def:Billboard}.

\begin{definition}[Billboard Database] \label{Def:Billboard}
A billboard database can be defined in the form of a tuple $<b_{id}, b_{loc}, b_{cost}>$, where $b_{id}$, $b_{loc}$ and $b_{cost}$ denotes a unique billboard id, billboard location, and billboard cost, respectively.
\end{definition}

Assume a set of digital billboards $\mathcal{B}=\{b_1, b_2, \ldots, b_m\}$ owned by an influence provider are placed across a city run for the duration $[T_1, T_2]$. Different commercial houses can hire these billboards slot-wise to show their advertisement content. Now, we define the notion of billboard slot in Definition \ref{Def:1}.

\begin{definition}[Billboard Slot] \label{Def:1}
A billboard slot is denoted by a tuple of the form $(b_j, [t,t+\Delta])$ where $b_j \in \mathcal{B}$ and $t \in \{T_1, T_1 + \Delta + 1, \ldots, T_2- \Delta-1\}$. Here, $\Delta$ denotes the duration of each slot.
\end{definition}
It can be observed that for each billboard, the number of slots associated with it will be $\frac{T_2-T_1}{\Delta}$ \footnote{Here, we assume that $\Delta$ perfectly divides $T_2-T_1$.}. There are $m$ many billboards $(i.e., |\mathcal{B}|=m)$, hence, the total number of slots denoted as $\mathcal{BS}$ are $m \cdot \frac{T_2-T_1}{\Delta}$. Depending on the available budget, an e-commerce house hires slots weekly or monthly to fulfill their influence demand. Next, we describe the influence function in Section \ref{Sec:Influence-Function}.

\subsection{Influence Function}\label{Sec:Influence-Function}
Normally, e-commerce companies approach the influence provider, who owns multiple billboard slots, to advertise their products and maximize the influence of their products. Now, one question arises: How is the influence of a billboard slot calculated? We state this in Definition \ref{Def:Influence-Function}.

\begin{definition}[Influence of a Billboard Slots]\label{Def:Influence-Function}
    Given a subset of billboard slots $\mathcal{S} \subseteq \mathcal{BS}$, the influence of $\mathcal{S}$  is denoted as $\mathcal{I}(\mathcal{S})$ and defined it as the sum of the influence probabilities of the individual users in the trajectory database.
    \begin{equation} \label{Eq:1}
    \mathcal{I}(\mathcal{S})= \underset{u \in \mathcal{D}}{\sum} [1-  \underset{b \in \mathcal{S}}{\prod}(1-Pr(u,b))]
    \end{equation}  
\end{definition}
Here, the influence function $\mathcal{I}()$ maps each subset of billboard slots to its corresponding influence value, i.e., $\mathcal{I}: 2^{\mathcal{BS}} \longrightarrow \mathbb{R}^{+}_{0}$ and $\mathcal{I}(\emptyset) = 0$. The influence function defined in Equation \ref{Eq:1} is widely used in the existing advertising literature \cite{ali2023influential,zhang2019optimizing,zhang2020towards,zhang2021minimizing} also.
\begin{lemma}
The Influence function $\mathcal{I}()$ is non-negative, monotone and submodular.
\end{lemma}

As mentioned previously, when one person views an advertisement running on a billboard slot, he is influenced by a certain probability. Most existing literature \cite{zhang2019optimizing,zhang2020towards,zhang2021minimizing,ali2023influential} does not consider that the probability value also depends on the running advertisement content (e.g., tag).
Therefore, we have considered tags to be an important factor in this paper. Consider $\mathcal{T}_{u_i}$ be the set of tags associated with user $u_{i}$ and for all the users in $\mathcal{U}$, the set of tags is relevant as a whole can be denoted as $\mathcal{T} = \underset{u_i \in \mathcal{U}}{\bigcup} \mathcal{T}_{u_i}$. Now, for every $u \in \mathcal{U}$ and $t \in \mathcal{T}$, the influence probability can be denoted as $Pr(u|t)$ and defined in Definition \ref{Def:Tagprobability}.

\begin{definition}[Tag Specific Influence Probability]\label{Def:Tagprobability}
Given a subset of tags $\mathcal{T}^{'} \subseteq \mathcal{T}$ and any user $u \in \mathcal{U}$, the tag-specific influence probability can be defined in Equation \ref{Eq:Tagspecificprobability}.

 \begin{equation} \label{Eq:Tagspecificprobability}
        Pr(u|\mathcal{T}^{'})= 1-\underset{t \in \mathcal{T}^{'}}{\prod} (1-Pr(u|t))
    \end{equation}
\end{definition}

Now, to find out the impact of tags in billboard slots, we define the Tag-specific influence of a billboard slot in Definition \ref{Def:TSIBS}.

\begin{definition}[Tag Specific Influence of Billboard slots] \label{Def:TSIBS}
Given a subset of billboard slots $\mathcal{S} \subseteq \mathcal{BS}$ and a subset of tags $\mathcal{T}^{'} \subseteq \mathcal{T}$, the tags-specific influence of $\mathcal{S}$ is denoted by $\mathcal{I}(\mathcal{S} | \mathcal{T}^{'})$ and defined using Equation No. \ref{Eq:TSIBS}.
\begin{equation} \label{Eq:TSIBS}
\mathcal{I}(\mathcal{S}|\mathcal{T}^{'})= \underset{u \in \mathcal{U}}{\sum} 1- \underset{b \in \mathcal{S}}{\prod} (1-Pr(u,b|\mathcal{T}^{'}))
\end{equation}
\end{definition}

Here, the influence function $\mathcal{I}()$ is a combined function that maps each tag and billboard slot to its corresponding influence value, i.e., $\mathcal{I}: 2^{\mathcal{T}} \times 2^{\mathcal{BS}} \longrightarrow \mathbb{R}^{+}_{0}$. Next, we describe the bipartite matching in Section \ref{Bipartite-Matching}.

\subsection{Bipartite Matching}\label{Bipartite-Matching}
In this work, we formulate our problem as a bipartite matching problem and state this problem in Section \ref{Problem-Definition}. Next, we define bipartite matching in Definition \ref{Def:Bipartite-Matching}.

\begin{definition}[Bipartite Matching]\label{Def:Bipartite-Matching}
A matching is bipartite if it contains a set of edges $\mathcal{E}_{b} \subseteq E$, which do not share any common vertex \cite{preis1999linear}.
\end{definition}


In traditional bipartite matching \cite{https://doi.org/10.1112/jlms/s1-10.37.26}, each tag is allocated to exactly one billboard slot, providing an optimal solution in polynomial time ($\mathcal{O}((n+m)^3)$ via the Hungarian Method) \cite{kuhn2005hungarian}. However, this doesn't suit our problem, as we need to allocate a single tag to multiple billboard slots, with each slot associated with only one tag. To address this, we introduce an iterative approach, formulating it as a one-to-many bipartite matching problem, defined in Definition \ref{Def:OMBM}.

\begin{definition}[One-to-Many Bipartite Matching]\label{Def:OMBM}
A matching is said to be one-to-many bipartite matching if it contains a set of edges $\mathcal{E}_{b} \in E$ that may share a common vertex in $u \in \mathcal{U}$ whereas each vertex in $v \in \mathcal{V}$ can connect with at most one vertex in $\mathcal{U}$ via an edge $e \in \mathcal{E}_{b}$ \cite{dutta2019one}.
\end{definition}

\subsection{Problem Definition}\label{Problem-Definition}
This section defines the Tag Allocation Problem formally. The inputs to this problem are trajectory, and billboard database, and set of selected slots ($\mathcal{S}$) and tags ($\mathcal{T}^{'}$) and the goal is allocate the tags to the slots such that the influence is maximized. One thing to highlight here is that two tags can not be allocated to a single slot. However, the vice versa can happen. Also, it can be observed that in the worst case the number of possible allocations will be of $\mathcal{O}(|\mathcal{S}|^{|\mathcal{T}^{'}|})$. Now, we state our problem formally in Definition \ref{Def:Problem_Statement}.


\begin{definition}[Tag Allocation problem]  \label{Def:Problem_Statement}
Given a trajectory database $\mathcal{D}$, a billboard database $\mathcal{B}$, a selected subset of slots $\mathcal{S}$ and tags $\mathcal{T}^{'}$ respectively, the tag allocation problem asks to assign a tag to a slot such that the influence is maximized.
\end{definition}
From the computational point of view, this problem can be posed as follows:

\begin{tcolorbox}
\underline{\textsc{Tag Allocation Problem}} \\
\textbf{Input:} A trajectory ($\mathcal{D}$) and Billboard ($\mathcal{B}$) Database, A set of slots ($\mathcal{S}$) and tags ($\mathcal{T}$).

\textbf{Problem:} Find out an allocation of the tags to the slots such that the influence is maximized.

\end{tcolorbox}
By a reduction from the Set Cover Problem, we can show that the Tag Assignment Problem is NP-hard. This result has been presented in Theorem \ref{NP-hard}. Due to the space limitation, we are not able to give the whole reduction. 

\begin{theorem}\label{NP-hard} The 
\textsc{Tag Assignment Problem in Billboard Advertisement} is NP-hard and hard to approximate in a constant factor.
\end{theorem}
Next, we discuss the proposed solution methodologies in Section \ref{Sec:PSA}.

\section{Proposed Solution Approach}\label{Sec:PSA}

Our experiments are represented in two folded ways: first, by selecting influential slots and tags, and second, by allocating tags to the slots. To address the first goal, we adopt the stochastic greedy approach for influential slots and tags selection introduced by Ali et al. \cite{ali2024influential}. Next, we construct a weighted bipartite graph using the selected slots and tags, and it is described in Section \ref{ConstructionOfGraph}.

\subsection{Construction of the Weighted Bipartite Graph}\label{ConstructionOfGraph}
 In this formulation, we have billboard slots and tags represented in the form of a bipartite graph $\mathcal{G}_{b} = (\mathcal{U} \cup \mathcal{V},\mathcal{E}_{b})$, where $\mathcal{U}$ contains the set of tags while $\mathcal{V}$ contains a set of billboard slots and in practice, $|\mathcal{U}| < < |\mathcal{V}|$. We have adopted the following approach to construct the edge set of the bipartite graph.  For every tag $u \in \mathcal{U}$ and every slot $v \in \mathcal{V}$, we compute the influence (i.e., edge weight) of all the individual allocations using Equation No. \ref{Eq:TSIBS}. Next, we compute the mean weight $(\mu)$ and standard deviation of all the edge weight using $\mu = \frac{1}{|\mathcal{E}_{b}|} \sum_{e \in \mathcal{E}_{b}} \mathcal{W}(e)$ and $\sigma = \sqrt{\frac{1}{|\mathcal{E}_{b}|} \sum_{e \in \mathcal{E}_{b}} (\mathcal{W}(e) - \mu)^2}$, respectively. For each $e \in \mathcal{E}_{b}$ we compute $\theta$-score value $\mathcal{Z}(e) = \frac{\mathcal{W}(e) - \mu}{\sigma}$ and prune the edge if $\mathcal{Z}(e) < \theta$, where $\theta$ is the user-defined parameter (it may be -1, 0, etc.). Basically, we set a threshold, i.e., $\mu + \theta \cdot \sigma$, and if $\mathcal{W}(e) < \mu + \theta \cdot \sigma$, then prune that edge. This method ensures that only edges with weights significantly lower than the average are pruned.

Next, we introduce the one-to-many bipartite matching approach and describe it in Section \ref{Sec:OMBM} to address the second goal.

\subsection{One-to-Many Bipartite Matching (OMBM)}\label{Sec:OMBM}
The OMBM algorithm takes a bipartite graph containing edges between tags and billboard slots as input and returns an allocation of tags to the billboard slots in the form of a 1-D array. This array maps a billboard slot to a tag such that each is allocated to exactly one tag; however, one tag may be allocated to more than one slot. At first, in Line No $1$ to $3$, we initialize the array $\mathcal{Q}$ as empty for the slots $v \in \mathcal{V}$. In-Line No. $4$ \texttt{while} loop will execute till the billboard slots, $|\mathcal{V}|$ is empty and in Line No. $5$ to $6$, for each slot $v \in \mathcal{V}$ store the possible matches in the adjacent vertices of $v$ into $\mathcal{C}_{v}$ from tags set $\mathcal{U}$. Initially, $\mathcal{C}_{v}$ is initialized with the adjacent vertices of $v$ i.e., $\mathcal{A}(v)$. Before further discussions, we define the notion of a dominating edge in Definition \ref{Def:Dominating-set}.
\begin{definition}[Dominating Edge]\label{Def:Dominating-set}
In a bipartite graph, an edge $\mathcal{E}_{uv}$ is said to be dominating if $\mathcal{W}_{uv} > \mathcal{W}_{uv^{'}}$ and $\mathcal{W}_{uv} > \mathcal{W}_{u^{'}v}$ where $\mathcal{W}_{uv}$ denotes edge weight of tag $u$ to slot $v$ and $v^{'} \neq v$, $u^{'} \neq u$, $u \in \mathcal{U}$, $v \in \mathcal{V}$.
\end{definition}

\par Next in Line No. $7$ dominating edge set $E_{d}$ is initialized to empty set. In-Line No. $8$ to $12$ for each slot $v \in \mathcal{V}$ the $lc(v)$ function takes $v$ as input and returns an adjacent tag $u$ of $v$ that has maximum edge weight, i.e., the best tag for matching. If $v$ is also the best match for $lc(u)$, then such a vertex pair will be added to $E_{d}$. Now, in Line No. $11$ to $17$, Algorithm \ref{Algo:OMBMP} deals with the edges that are not added to $E_{d}$ in Line No. $8$~ to~ $12$. At first, we pick an end node $u \in \mathcal{U}$ from $E_{d}$. Next, we delete the edge $\mathcal{E}_{uv}$ associated with $u$ as the slot is already matched with a tag. In-Line No. $17$, we iterate through all $\mathcal{C}_{u}$, and this represents the slots that have not yet been matched with any tag. Now, in Line No. $18$, we also check if the total number of slots $|Count_{lc(v)}|$ connected for the tag $lc(v)$ is still within upper limit $Bound_{lc(v)}$ or not. If both conditions are satisfied, we add an edge with $v$ and $u$ to $E_{d}$ and allocate slot $v$ to tag $u$. In-Line No. $21$ and $22$, we remove the allocated tag $u$ and slot $v$, respectively. Finally, Algorithm \ref{Algo:OMBMP} will return $\mathcal{Q}$ by allocating billboard slots to tags.
\SetKwComment{Comment}{/* }{ */}
\begin{algorithm}[H]
\scriptsize
 \KwData{The Trajectory Database $\mathcal{D}$, The Billboard Database $\mathcal{B}$, Tag Database $\mathbb{T}$, A weighted Bipartite Graph $\mathcal{G}(\{\mathcal{U} \cup \mathcal{V}\}, E, \mathcal{W})$.}
 \KwResult{An allocation of Tags to the billboard slots to maximize influence.}
Initialize a 1-D array $\mathcal{Q}$, $\mathcal{U}^{'} = \mathcal{U}$\;
\For{$\text{ each }v \in \mathcal{V}$}{
 $\mathcal{Q}(v) \leftarrow \emptyset$ \Comment*[r]{$\mathcal{Q}$ is the 1-D array}
 }
\While{$|\mathcal{V}| > 0$}{
\For{$\text{ each }v \in \mathcal{V}$}{
 $C_{v} \leftarrow \mathcal{A}(v)$ \Comment*[r]{$C_{v}$ is the candidate and $\mathcal{A}_{v}$ is the adjacent vertices}
 }
$E_{d} \leftarrow \emptyset$ \Comment*[r]{$E_{d}$ is the dominating edge set}

\For{$\text{ each }v \in \mathcal{V}$}{
 $u \leftarrow lc(v)$ \Comment*[r]{$lc()$ find best matching vertex}
\If{$lc(u) = v$}{
$E_{d} = E_{d} \cup \{v, lc(v)\}$\;
$\mathcal{Q}(v) = lc(v)$\;
}}

\While{$E_{d} \neq \emptyset$}{
$u \leftarrow \text{end vertex} \in \mathcal{U}~ \text{from an edge in} ~E_{d}$\;
$v \leftarrow \text{end vertex} \in \mathcal{V}~ \text{from an edge in} ~E_{d}$\;
$E_{d} \leftarrow E_{d} \setminus \{ v^{'},u\}~ \text{where}~ \mathcal{Q}(v^{'}) = u$\;
\For{each $v \in C_{u}$ where $\mathcal{Q}(v) = \emptyset$}{
\If{$lc(lc(v)) = v$~ and~ $|count_{lc(v)}| < Bound_{lc(v)}$}{
$E_{d} = E_{d} \cup \{v, lc(v)\}$\;
$\mathcal{Q}(v) = lc(v)$\;
}}
$\mathcal{U} \leftarrow \mathcal{U} \setminus lc(v)$\;
$\mathcal{V} \leftarrow \mathcal{V} \setminus v$\;
}
\If{$\mathcal{U} = \emptyset$}{
$\mathcal{U} = \mathcal{U}^{'}$\;
}
}
return $\mathcal{Q}$\;
\caption{OMBM Algorithm for Tag Allocation Problem}
\label{Algo:OMBMP}
\end{algorithm}

\paragraph{\textbf{Complexity Analysis.}}
In-Line No. $1$ to $3$, initializing a 1-dimensional array will take $\mathcal{O}(k)$ time as there is $k$ number of slots. The \texttt{while} loop at Line No. $4$ will execute for $\mathcal{O}(k)$ times and Line No. $5$ to $6$ will execute for $\mathcal{O}(k^{2})$. Next, in Line No. $7$ initializing dominating set $E_{d}$ will take $\mathcal{O}(k)$ time. Line No. $8$ \texttt{for loop} will execute for $\mathcal{O}(k^{2})$ times, and Line No. $9$ finding the best-matching vertex will take $\mathcal{O}(k^{2} \cdot \ell)$ time. Line No. $10$ to $12$ will take $\mathcal{O}(k^{2})$ times, and Line No. $8$ to $12$ will take $\mathcal{O}(k^{2} \cdot \ell + k^{2})$ time to execute. Next, in Line No. $13$ \texttt{while loop} will execute for $\mathcal{O}(k^{2})$ time and Line No. $14$ to $16$ will take $\mathcal{O}(k^{2})$ time. The \texttt{for loop} at Line No. $17$ to $20$ will execute for $\mathcal{O}(k^{2} \cdot \ell)$ time. Line No. $21$ and $22$ will take $\mathcal{O}(k^{2})$ time. Finally, Line No. $23$ to $24$ will be executed for $\mathcal{O}(k)$ times in the worst case. So, Algorithm \ref{Algo:OMBMP} will take total $\mathcal{O}(k + k + k^{2} \cdot \ell + k^{2} + k^{2} \cdot \ell)$, i.e., $\mathcal{O}(k^{2} \cdot \ell)$ time to execute in the worst case. The additional space requirement for Algorithm \ref{Algo:OMBMP} will be $\mathcal{O}(k + k + \ell)$ i.e., $\mathcal{O}(k + \ell)$ for the 1-D array $\mathcal{Q}$, $\mathcal{C}(v)$ and $E_{d}$. 

\par Note that $|E| = k\ell$ in the case of a complete bipartite graph. However, if some slots are not eligible for certain tags to be allocated due to the $\theta$-score threshold value, then $|E| < k\ell$. Therefore, we state the complexity of Algorithm \ref{Algo:OMBMP} as linear on the number of edges generated between tags and slots.

\begin{theorem}
The time and space complexity of the Algorithm \ref{Algo:OMBMP} will be $\mathcal{O}(k^{2} \cdot \ell)$ and $\mathcal{O}(k+\ell)$, respectively.
\end{theorem}

\begin{lemma}\label{Lemma1}
All edges in the dominating edge set $E_{d}$ must be a part of the optimal matching in the solution obtained before Line No. $13$ in Algorithm \ref{Algo:OMBMP}.
\end{lemma}

\begin{proof}
 Assume \(\mathcal{M}^{opt}\) is an optimal matching in graph \(\mathcal{G}\). Suppose there exists an edge \(\{u,v\} \in E_{d}\) not in \(\mathcal{M}^{opt}\). Let \(\{u,v'\} \in \mathcal{M}^{opt}\) with \(\mathcal{W}_{uv} > \mathcal{W}_{uv'}\). Consider \(\mathcal{M}^{'} = [\mathcal{M}^{opt} \setminus \{u,v'\}] \cup \{u,v\}\). Then \(\mathcal{W}(\mathcal{M}^{'}) > \mathcal{W}(\mathcal{M}^{opt})\), contradicting \(\mathcal{M}^{opt}\)'s optimality. Hence, all edges in \(E_{d}\) are in \(\mathcal{M}^{opt}\).
\end{proof}

\begin{lemma}\label{Lemma2}
Each billboard slot will be assigned at most one tag in the solution.
\end{lemma}

\begin{proof}
Given a billboard slot \( v \in \mathcal{V} \) and a tag $u \in \mathcal{U}$, \( u \) is assigned to \( v \) only if \( \mathcal{Q}(v) = \emptyset \) (as stated in Line No. $17$ to $20$ of Algorithm \ref{Algo:OMBMP}). Once a tag \( u \) is assigned to a billboard slot \( v \), \( \mathcal{Q}(v) \) is set to \( u \). Since \( \mathcal{Q}(v) \neq \emptyset \) after this assignment, the condition \( \mathcal{Q}(v) = \emptyset \) will no longer hold for \( v \). Therefore, \( v \) cannot be reassigned to another tag and it defines $\forall ~v \in \mathcal{V}, \exists \text{ at most one } u \in \mathcal{U} \text{ such that } \mathcal{Q}(v) = u.$
This implies that each billboard slot \( v \) is assigned at most one tag in the solution.
\end{proof}

\begin{lemma}\label{Lemma3}
Each tag $ u_{i} \in \mathcal{U}$ will finally have exactly $Bound_{i}$ many billboard slots assigned to it.
\end{lemma}

\begin{proof}
In Algorithm \ref{Algo:OMBMP}, the allocation performed in Line No. $18$ $\mathcal{Q}_{v} = lc(v)$ happens when $|Count_{lc(v)}| < Bound_{lc(v)}$ and along with Lemma \ref{Lemma2} allows us to conclude this.
\end{proof}

\subsection{Tag Allocation Process}
The overall tag allocation procedure is shown in Algorithm \ref{Algo:TAP}. At first, in Line No. $1$, a set of billboard slots $(\mathcal{S})$ and a set of tags $(\mathcal{T})$ is provided as an input in the stochastic greedy \cite{ali2024influential}, and it returns a subset of slots and tags, i.e., $\mathcal{S}^{'} \subseteq \mathcal{S}$ and $\mathcal{T}^{'} \subseteq \mathcal{T}$ as output. Next, in Line No. $2$, the bipartite graph is generated using bipartite sets of tags and billboard slots. In-Line No. $3$, based on the $\theta$-score threshold value, a pruned bipartite graph $(\mathcal{G}^{'})$ is generated. In-Line No $4$, $\mathcal{G}^{'}$ is used in Algorithm \ref{Algo:OMBMP} as input, where it finds a mapping of the billboard slots to the tags and returns a multi-slot tag allocation.

\SetKwComment{Comment}{/* }{ */}
\begin{algorithm}
\scriptsize
 \KwData{The Trajectory Database $\mathcal{D}$, The Billboard Database $\mathcal{B}$, Tag Database $\mathbb{T}$, Context Specific Influence Probabilities, A set of Billboard Slots $(\mathcal{S})$, A set of Tags $(\mathcal{T})$, Two Positive Integers $k$ and $\ell$.}
\KwResult{An allocation of Tags to the billboard slots to maximize influence.}
$\mathcal{S}^{'}, \mathcal{T}^{'} \longleftarrow StochasticGreedy(\mathcal{S}, \mathcal{T})$ \cite{ali2024influential} \Comment*[r]{$\mathcal{S}^{'} \subseteq \mathcal{S}$ with $|\mathcal{S}^{'}|=k$ and $\mathcal{T}^{'} \subseteq \mathcal{T}$ with $|\mathcal{T}^{'}|=\ell$}
$\mathcal{G}(\mathcal{S^{'}, \mathcal{T}^{'}}, \mathcal{E}_{b}) \longleftarrow \text{Generate a Bipartite Graph}$\;
$\mathcal{G}^{'}(\mathcal{S^{''}, \mathcal{T}^{''}}, \mathcal{E}_{b}^{'}) \longleftarrow \text{ Generate Pruned Bipartite Graph}$\;
$\mathcal{Q} \longleftarrow OMBM(\mathcal{G}^{'})$ \Comment*[r]{Algorithm \ref{Algo:OMBMP}}
return $\mathcal{Q}$\;
 \caption{Tag Allocation Algorithm for Multi-slot Tag Allocation}
 \label{Algo:TAP}
\end{algorithm}
\vspace{-0.28cm}

\begin{theorem}
Algorithm \ref{Algo:OMBMP} gives an approximation ratio of $\rho \leq 1 + \max_{1 \leq i \leq m} (\mathcal{K}_i - \delta_i)$, where $\delta_{i} \in \{0,1\}$ and $\mathcal{K}_{i}$ is the set of slots assigned to tag $u_{i}$.
\end{theorem}
\begin{proof}
As mentioned previously, in Algorithm \ref{Algo:OMBMP} Line No. $13$ to $20$ performs a matching of tag \(u_i\) to \((|\mathcal{K}_i| - \delta_i)\) slots where $\delta_{i} \in \{0,1\}$~ for $i = 1, 2, \ldots, \ell$. Next, $lc(v) = u$ represents the tag $u_{i}$ is matched with slot $v_{j}$ where $v \in \mathcal{V}$ as given in Line No. $18$ $v = lc(u)= lc(lc(v))$ is satisfied. If $\{v,u\}$ is not part of optimal matching, however, it is added to the $E_{d}$ then at most $(1 + \mathcal{K}_{i} - \delta_{i})$ many edges may not be considered. The cost for assigning slots to a tag \(t_i\) will be at most \(|\mathcal{K}_i| + 1 - \delta_i\). So, the total cost \(C_{\text{A}}\) can be approximated by summing the costs over all tags \(t_i\) i.e., $C_{\text{A}} \leq \sum_{i=1}^m (|\mathcal{K}_i| + 1 - \delta_i)$ and the optimal solution has a cost \(C_{\text{opt}} = \sum_{i=1}^m |\mathcal{K}_i|\). So we can write $\rho = \max \left( \frac{\sum_{i=1}^m (|\mathcal{K}_i| + 1 - \delta_i)}{\sum_{i=1}^m |\mathcal{K}_i|} \right)$.
Since \(|\mathcal{K}_i|\) are the optimal slots, $\rho \leq 1 + \max_{1 \leq i \leq m} \left( \frac{1 - \delta_i}{|\mathcal{K}_i|} \right)$. The term \(\max_{1 \leq i \leq m} \left( \frac{1 - \delta_i}{|\mathcal{K}_i|} \right)\) simplifies to \((1 - \delta_i)\) for the worst-case scenario where \(|\mathcal{K}_i| = 1\) and hence $\rho \leq 1 + \max_{1 \leq i \leq m} (\mathcal{K}_i - \delta_i)$.
\end{proof}

\paragraph{\textbf{An Illustrative Example}}\label{Example}
 Consider there are ten billboard slots $\mathcal{V} = \{b_{0},b_{1},\ldots, $ $ b_{9}\}$, three tags $\mathcal{U} = \{t_{0}, t_{1}, t_{2}\}$ with corresponding influence value as output after applying stochastic greedy approach \cite{ali2024influential} on the billboard, and tag datasets as shown in Table \ref{ETable:1}, \ref{ETable:2}. Next, a bipartite graph $\mathcal{G}$ and its corresponding $10 \times 3$ weight matrix is generated. Now, based on the $\theta$-score value described in Section \ref{ConstructionOfGraph}, a new bipartite graph $\mathcal{G}^{'}$ with an updated weight matrix is generated. We initialize an empty $(Initialize~ to  ~-1)$ 1-D array $\mathcal{Q}$ to store the allocation as shown in Table \ref{ETable:1D_Array}. In the first iteration, the best allocations for the slots to the tags are as follows: $\{(0 \rightarrow 2), (1 \rightarrow 2), (3 \rightarrow 2), (4 \rightarrow 2), (5 \rightarrow 2), (6 \rightarrow 2), (7 \rightarrow 2), (8 \rightarrow 2), (9 \rightarrow 2)\}$, i.e., all the slots are allocated to only one tag $t_{2}$ and the allocation for the tag to the slot are $\{ (0 \rightarrow 4), (1 \rightarrow 4), (2 \rightarrow 4)\}$, i.e., slot $b_{4}$ is the best match for $t_{0}, t_{1}$ and $t_{2}$ and $b_{4}$ is matched to tag $t_{2}$ as shown in Table \ref{ETable:updated-1D_Array}. The edge weight of all the edges from slot $b_{4}$ to the tags is set to be $0$, and the influence of tag $t_{2}$ is set to $0$. Next, the matching is performed between tags $t_{0}, t_{1}$ and all the slots except $b_{4}$. Similarly, slots $b_{9}$ matched with $t_{1}$ and $b_{2}$ is matched with $t_{0}$. After this, all the tag's influence is $0$, iteration $1$ is completed, and updated allocation is shown in Table \ref{ETable:updated-1D_Array}. Still, some slots are unallocated. So, assign the initial influence value for all the tags and repeat the process till improvement in allocation occurs. The final allocation of slots to tags is as follows: $\{(2 \rightarrow 0), (4 \rightarrow 2), (5 \rightarrow 0), (7 \rightarrow 2), (8 \rightarrow 1), (9 \rightarrow 1)\}$ as shown in Table \ref{ETable:Final-1D_Array}. 

\begin{table}[!h]
\begin{center}
\begin{minipage}{0.5\textwidth}
\begin{center}
   \begin{tabular}{| c | c | c | c | c | c | c | c | c | c | c |}
   \hline
   $\mathcal{V}$ & $b_{0}$ & $b_{1}$ & $b_{2}$ & $b_{3}$ & $b_{4}$ & $b_{5}$ & $b_{6}$ & $b_{7}$ & $b_{8}$ & $b_{9}$  \\ \hline
   $I(b_{i})$ & 0.1 & 0.2 & 0.7 & 0.1 & 0.9 & 0.4 & 0.1 & 0.5 & 0.45 & 0.7 \\ \hline
   \end{tabular}
   \caption{\label{ETable:1} Billboard Influence}
\end{center}
\end{minipage}\hfill
\begin{minipage}{0.5\textwidth}
\begin{center}
   \begin{tabular}{| c | c | c | c |}
   \hline
   $\mathcal{U}$ & $t_{0}$ & $t_{1}$ & $t_{2}$ \\ \hline
   $I(b_{i})$ & 0.1 & 0.4 & 0.5 \\ \hline
   \end{tabular}
   \caption{\label{ETable:2} Tag Influence}
\end{center}
\end{minipage}
\end{center}
\end{table}


\begin{table}[!h]
\begin{center}
   \begin{tabular}{| c | c | c | c | c | c | c | c | c | c | c |}
   \hline
   $\mathcal{Q}$ & $-1$ & $-1$ & $-1$ & $-1$ & $-1$ & $-1$ & $-1$ & $-1$ & $-1$ & $-1$ \\ \hline
   \end{tabular}
   \caption{\label{ETable:1D_Array} Initial Allocation (1-D Array)}
\end{center}
\end{table}

\begin{table}[!h]
\begin{center}
\begin{minipage}{0.5\textwidth}
\begin{center}
   \begin{tabular}{| c | c | c | c | c | c | c | c | c | c | c |}
   \hline
   $\mathcal{Q}$ & $-1$ & $-1$ & $-1$ & $-1$ & $2$ & $-1$ & $-1$ & $-1$ & $-1$ & $-1$ \\ \hline
   \end{tabular}
   \caption{\label{ETable:updated-1D_Array} Updated Allocation (After $1^{st}$ iteration)}
\end{center}
\end{minipage}\hfill
\begin{minipage}{0.5\textwidth}
\begin{center}
   \begin{tabular}{| c | c | c | c | c | c | c | c | c | c | c |}
   \hline
   $\mathcal{Q}$ & $-1$ & $-1$ & $0$ & $-1$ & $2$ & $0$ & $-1$ & $2$ & $1$ & $1$ \\ \hline
   \end{tabular}
   \caption{\label{ETable:Final-1D_Array} Final Allocation}
\end{center}
\end{minipage}
\end{center}
\end{table}
\vspace{-1.99cm}

\section{Experimental Evaluation}\label{Sec:Experimental_Evaluation}
In this section, we describe the experimental evaluation of the proposed solution approach. Initially, we start by describing the datasets. 
\paragraph{\textbf{Dataset Description.}}
We use two datasets for our experiments, previously utilized in various trajectory data analytics studies \cite{yang2014modeling, ali2023influential,zhang2020towards}. The New York City (NYC) Dataset, collected from April 12, 2012, to February 16, 2013, includes 227,428 check-ins containing user ID, location name, timestamps, and GPS coordinates \footnote{\url{https://www.nyc.gov/site/tlc/about/tlc-trip-record-data.page}}. The VehDS-LA (Vehicle Dataset in Los Angeles) consists of 74,170 samples from 15 streets, featuring data like user ID, street name, latitude, longitude, and timestamps \footnote{\url{https://github.com/Ibtihal-Alablani}}. Further, we have created tag datasets from trajectory datasets of NYC and VehDS-LA containing tag names and tag influences to do our experiments. Additionally, we use billboard data from LAMAR \footnote{\url{http://www.lamar.com/InventoryBrowser}}, including billboard ID, latitude, longitude, timestamp, and panel size. The NYC dataset contains 716 billboards (1031040 slots), and the LA dataset includes 1483 billboards (2135520 slots).

\paragraph{\textbf{Experimental Setup.}} All the key parameters used in our experiments are summarized in Table \ref{Key-parameters}, and the default settings are highlighted in bold. The parameters $k$ and $\ell$ denote the number of slots and tags, respectively, whereas $\epsilon$ decides the sample set size in the stochastic greedy \cite{ali2024influential} from which maximum influential slots and tags are chosen. The user-defined parameter $\theta$ controls the $\theta$-score threshold, and $\lambda$ represents the maximum distance a slot can influence trajectories. 
\begin{table}[h!]
\caption{\label{Key-parameters} Key Parameters}
\vspace{-0.15 in}
\begin{center}
    \begin{tabular}{ | p{2cm}| p{5.5cm}|}
    \hline
    Parameter & Values  \\ \hline
    $k$ & $100, 150, 200, 250, \textbf{300}$   \\ \hline
    $\ell$ & $25, 50, 75, \textbf{100}, 125$   \\ \hline
    $\epsilon$ & $\textbf{0.01}, 0.05, 0.1, 0.15, 0.2$ \\ \hline
    $\theta$ & -2, $\textbf{-1}, 0, 1, 2$  \\ \hline
    $\lambda$ & $25m,50m,\textbf{100m},125m,150m$  \\ \hline
    \end{tabular}
\end{center}
\end{table}

\par All codes are implemented in Python and executed on an HP Z4 workstation with 64 GB memory and an Xeon(R) 3.50 GHz processor. All the proposed and baseline methods are demonstrated for their effectiveness and efficiency. All codes are executed five times, and average results are reported. Due to the anonymity constraint, we are not putting the Github Link of our implementations which we will do during camera ready submission.

\paragraph{\textbf{Goals of our Experiments.}} \label{Sec:Research_Questions}
In this study, we address the following Research Questions (RQ).
\begin{itemize}
\item \textbf{RQ1}: How does the number of matching tags to slots increase if we increase the number of slots and tags selected?
\item \textbf{RQ2}: If we increase the number of slots and tags, how do the computational time requirements of the proposed methods change? 
\item \textbf{RQ3}: Varying $\theta$-score, how does the number of matched tags and slots vary?
\item \textbf{RQ4}: How does the influence vary before and after allocating tags to slots?
\end{itemize}

\paragraph{\textbf{Baseline Methods.}}
 One tag can be assigned to multiple slots in all the baseline methods; however, the opposite is restricted. Now, we will discuss different baseline methods to compare with our proposed approach as follows:

\begin{itemize}
\item \textbf{Brute-force Method (BM).} In this approach, for each billboard slot, we search for the best match in each tag based on the edge weight (e.g., influence value), and tags are assigned to slots.
\item \textbf{Max Degree Allocation (MDA).} At first, tags are sorted in descending order based on the maximum degree. Then, for each slot, check the degree of the tags connected to the slot and the maximum degree tag allocated to that slot.
\item \textbf{Top-$k$ Slot and Random Tag (TSRT).} First, each slot is sorted in descending order based on individual influence value. Next, tags are allocated to the sorted slots randomly.
\item \textbf{Random Allocation (RA).} In this approach, tags are allocated to slots uniformly at random.
\end{itemize}

\paragraph{\textbf{Experimental Results and Discussions.}}
This section will discuss the experimental results of the solution methodologies
and address the research questions mentioned in Section \ref{Sec:Research_Questions}.

\paragraph{Tags, Slots Vs. No. of Matching.}
Figure \ref{Fig:Plot}(a,b) and \ref{Fig:Plot}(g,h) show the effectiveness of the baseline and proposed methods in the NYC and LA datasets, respectively. We have three main observations. With the fixed number of tags $(100)$, when the number of slots varies from $100$ to $300$, the number of matched tags and slots increase for baseline and proposed methods, as shown in Figure \ref{Fig:Plot}(a,g) and \ref{Fig:Plot}(b,h). Second, when we increase the number of slots in the `BM' and `MDA' approaches, all the slots are matched with only one tag, and this happens because there is one most influential tag that has a maximum degree in the bipartite graph. Among the baseline methods, `RA' and `TSRT' perform well. The `OMBM' has a smaller number of matched slots because in every iteration, a tag can only be allocated to a slot if the slot is the best match for a tag, and that tag is also the best for the slot. The `OMBM' approach in the LA dataset has fewer matched slots than `RA' and `TSRT' by $2.29 \%$ to $3\%$ and $0.76\%$ to $2\%$, respectively. In the NYC dataset, fewer matched slots are around $10.70\%$ compared to the baseline methods. Third, in the NYC dataset, the `OMBM' outperforms all the baseline methods regarding the number of matched tags, while in the LA dataset, the number of matched tags in `OMBM' is less. This happens because in NYC, the $\theta$-score threshold value is very low compared to the LA dataset, and for this reason, the number of edges that remain in the bipartite graph is huge, and more tags are matched.

\begin{figure*}[!ht]
\centering
\begin{tabular}{ccc}
\includegraphics[scale=0.17]{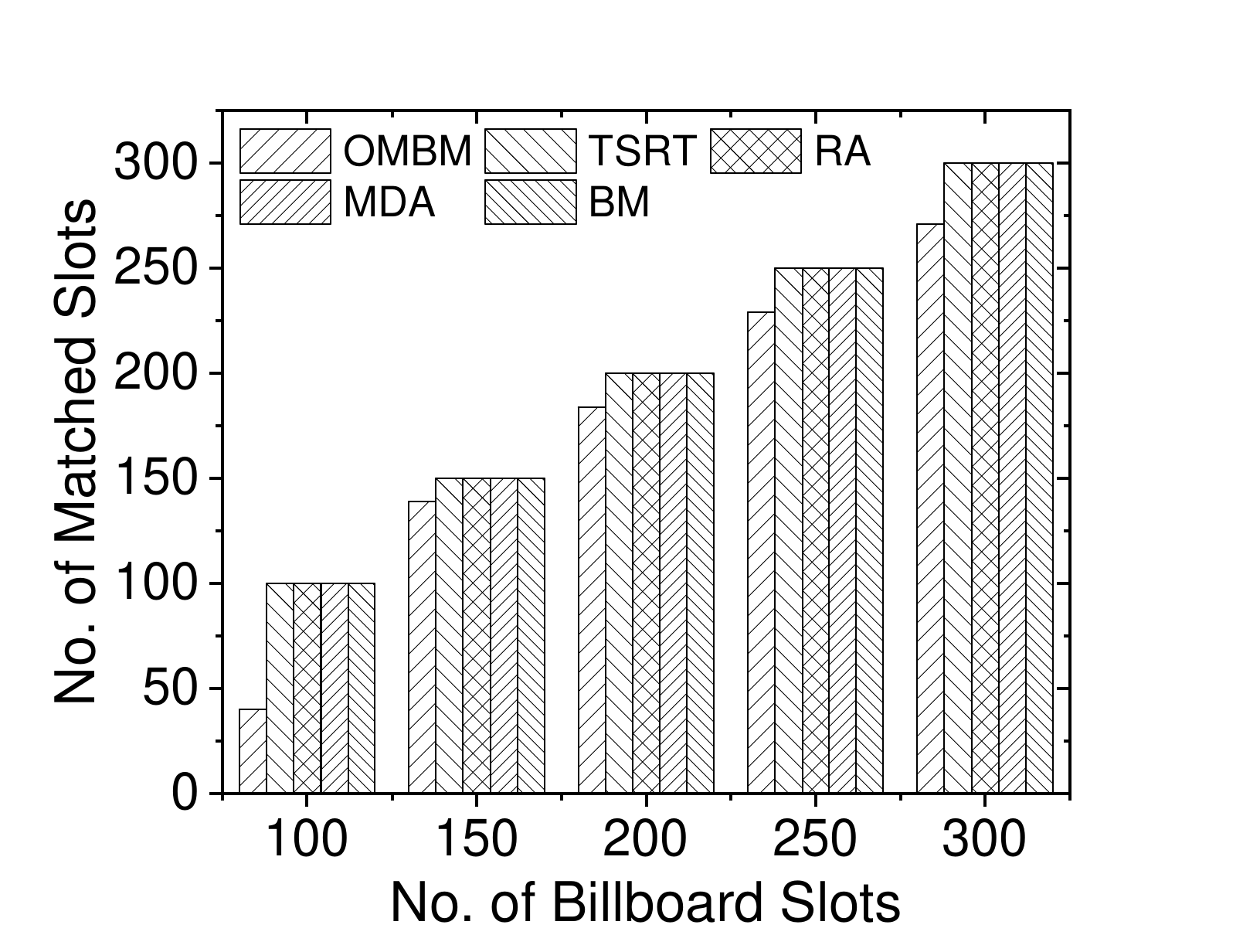} & \includegraphics[scale=0.17]{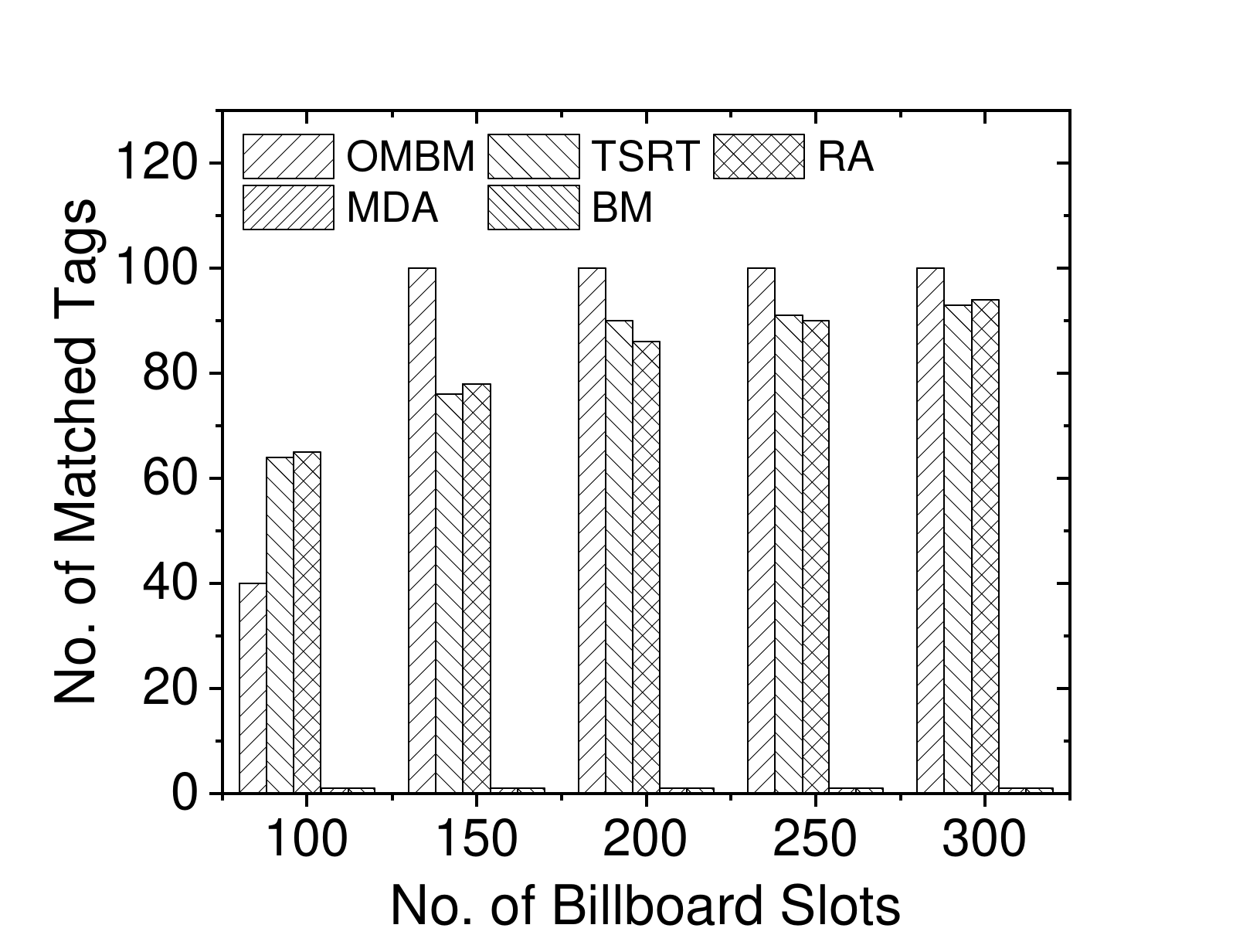} & \includegraphics[scale=0.17]{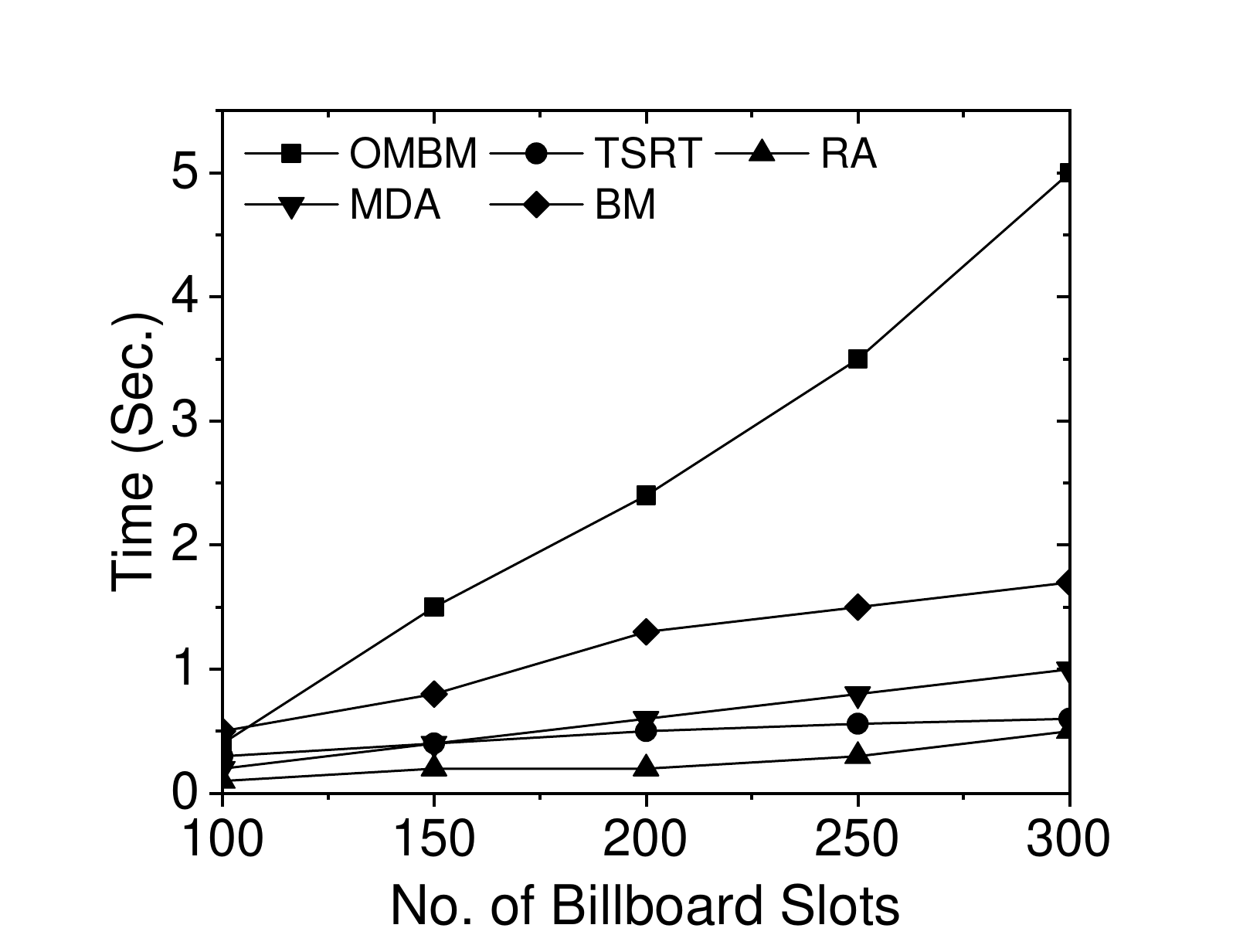} \\
(a) Matched Slots in NYC &  (b) Matched Tags in NYC & (c) Runtime in NYC  \\
\includegraphics[scale=0.17]{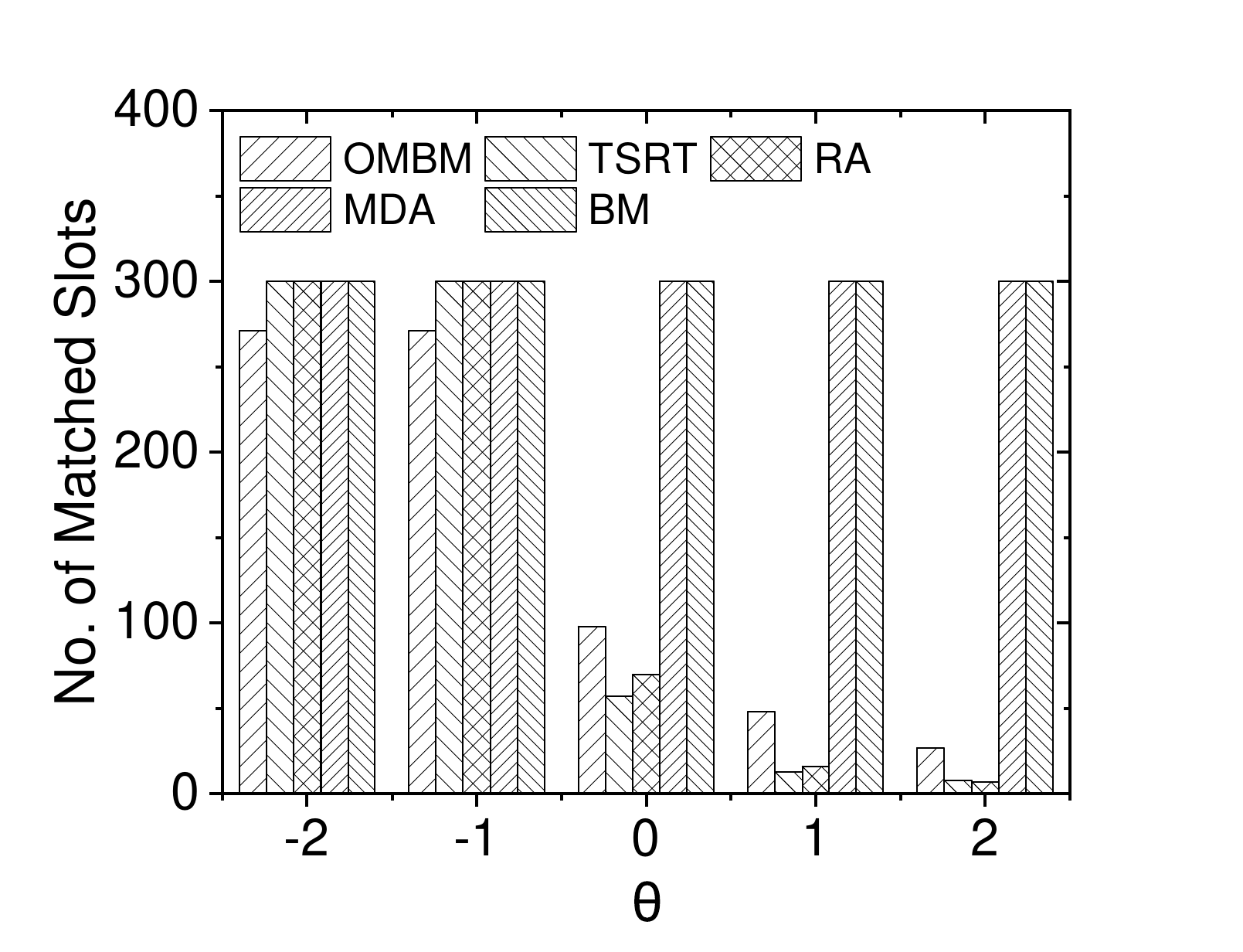} & \includegraphics[scale=0.17]{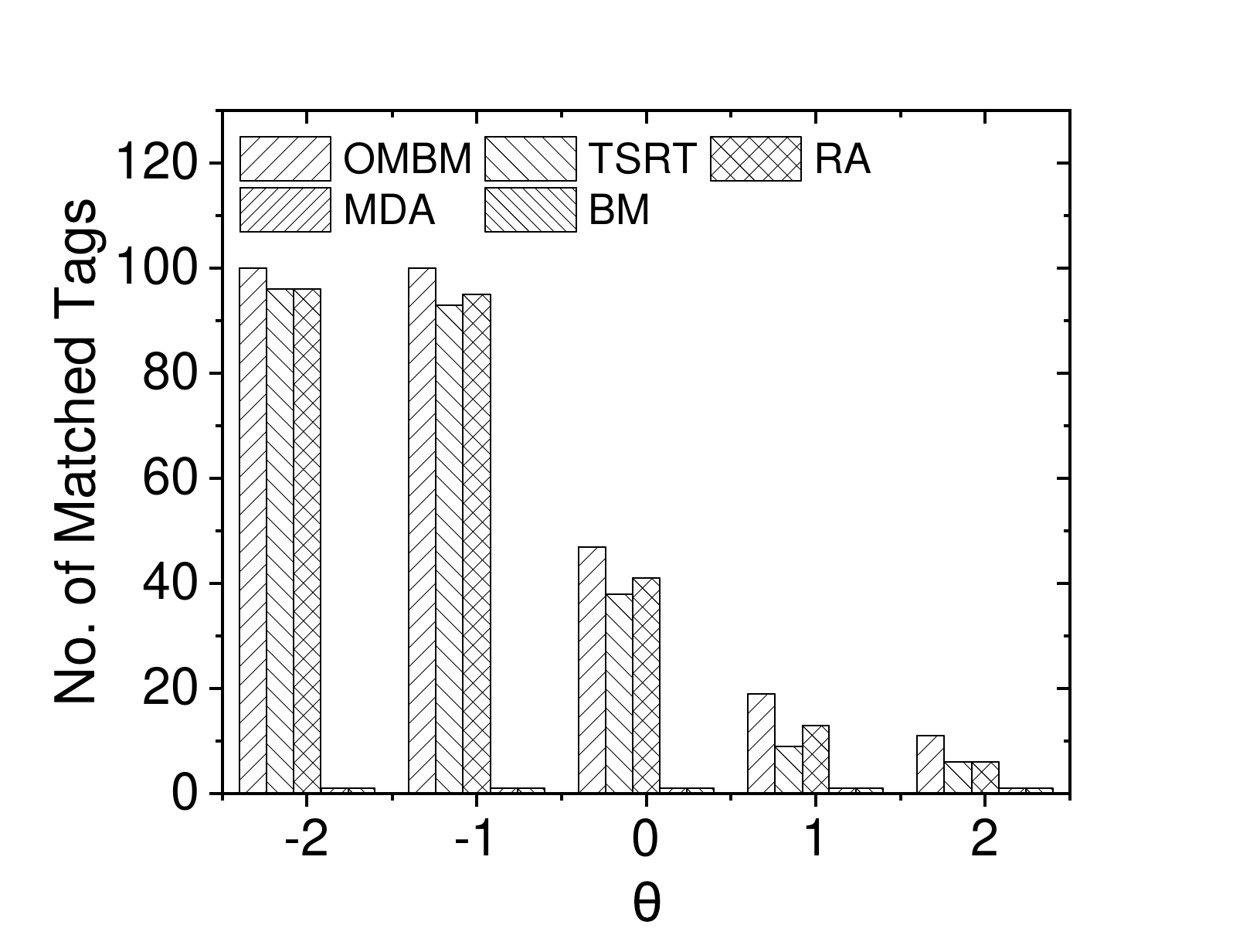} & \includegraphics[scale=0.17]{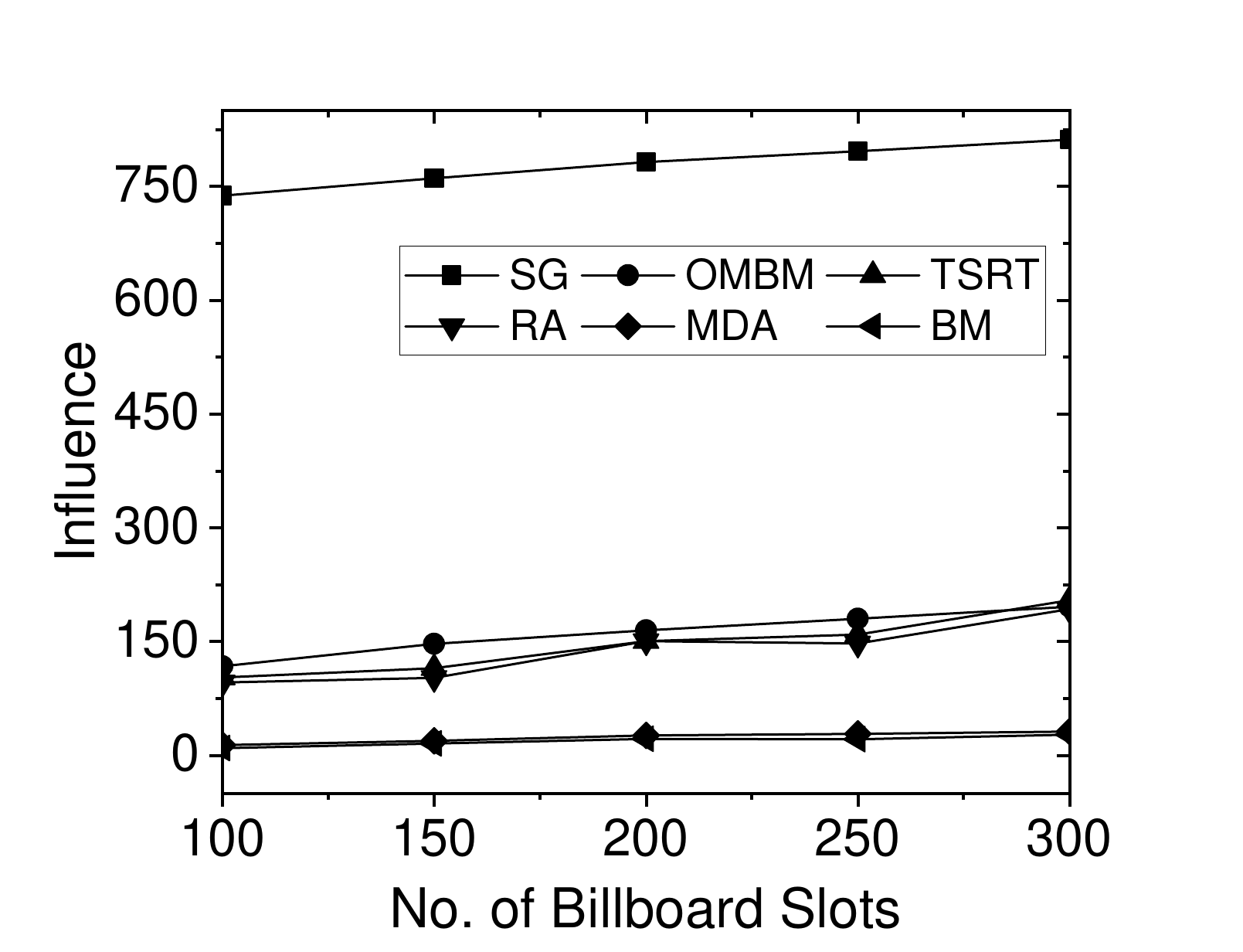}  \\
(d) Varying $\theta$ Vs. Slots (NYC) & (e)  Varying $\theta$ Vs. Tags (NYC)  & (f) Slots Vs. Influence (NYC) \\
\includegraphics[scale=0.17]{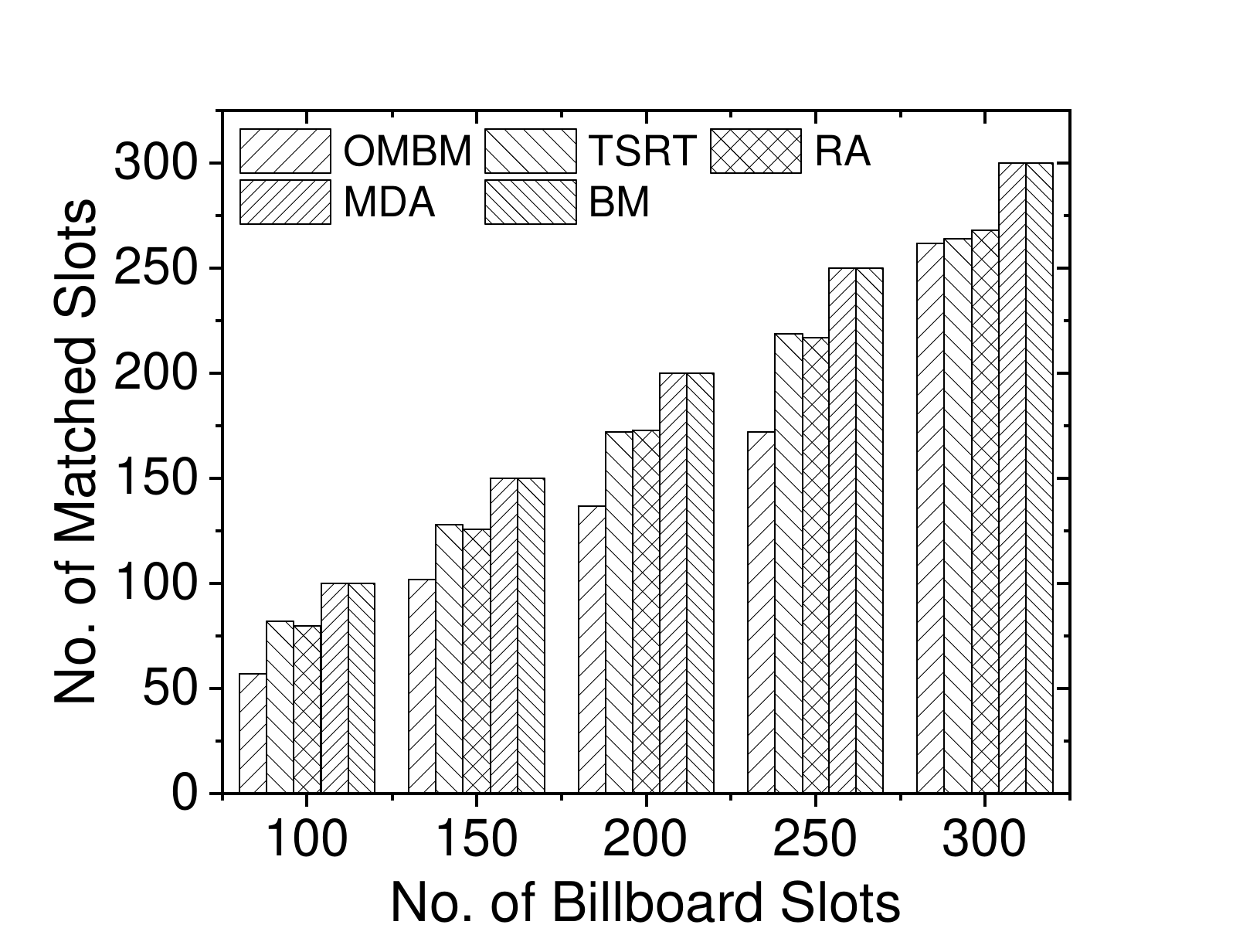} & \includegraphics[scale=0.17]{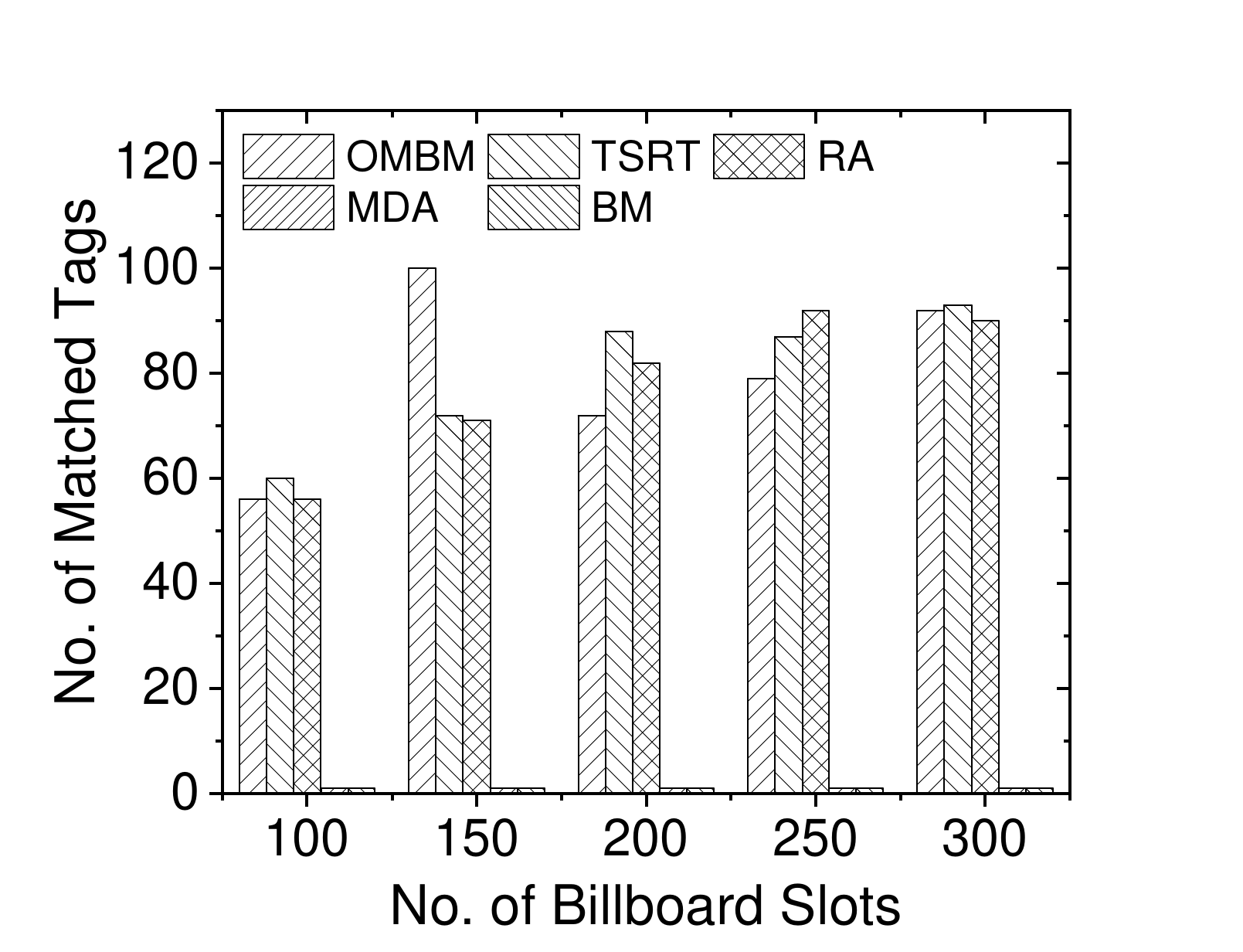}  & \includegraphics[scale=0.17]{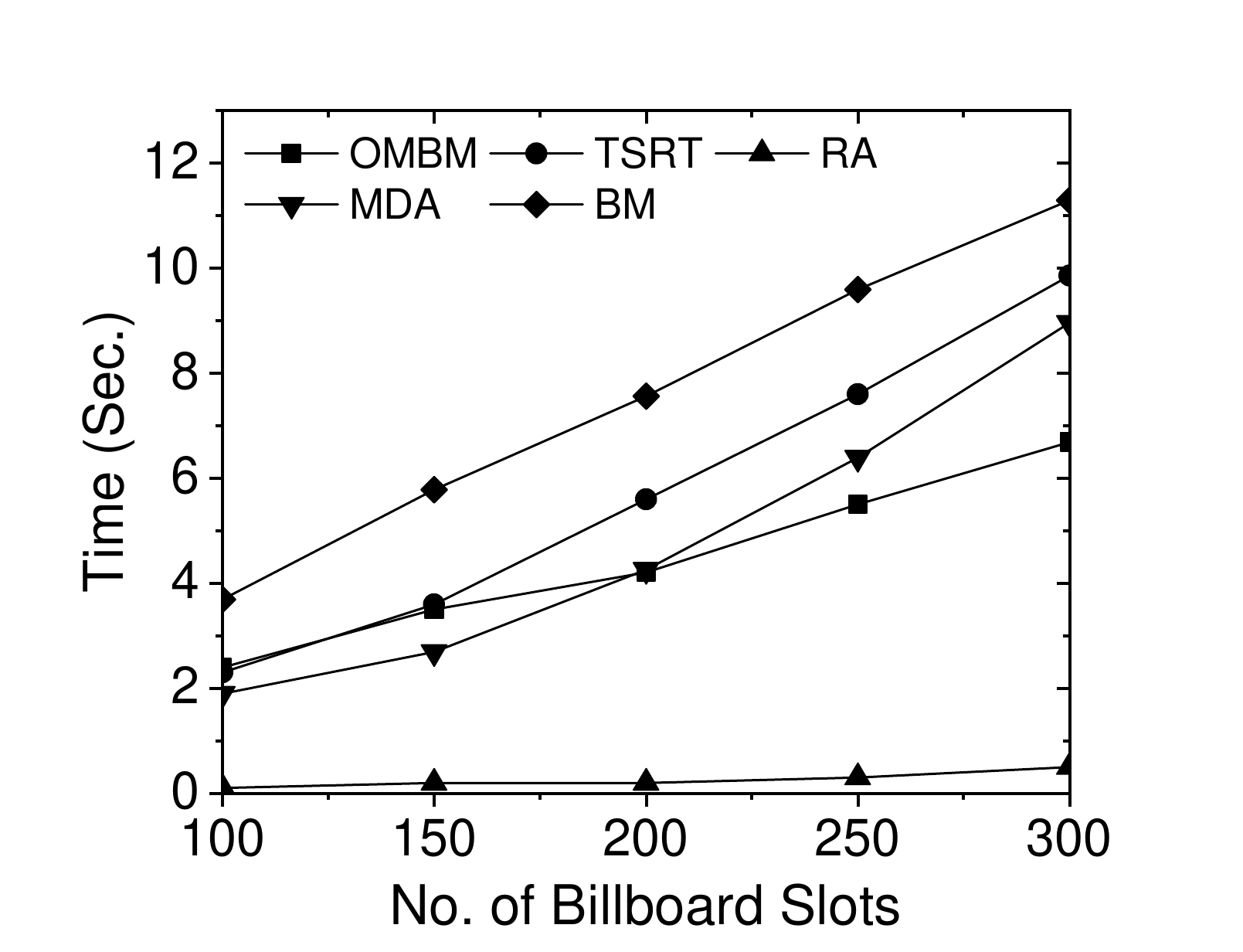} \\
(g) Matched Slots in LA & (h) Matched Tags in LA & (i) Runtime in LA \\
\includegraphics[scale=0.17]{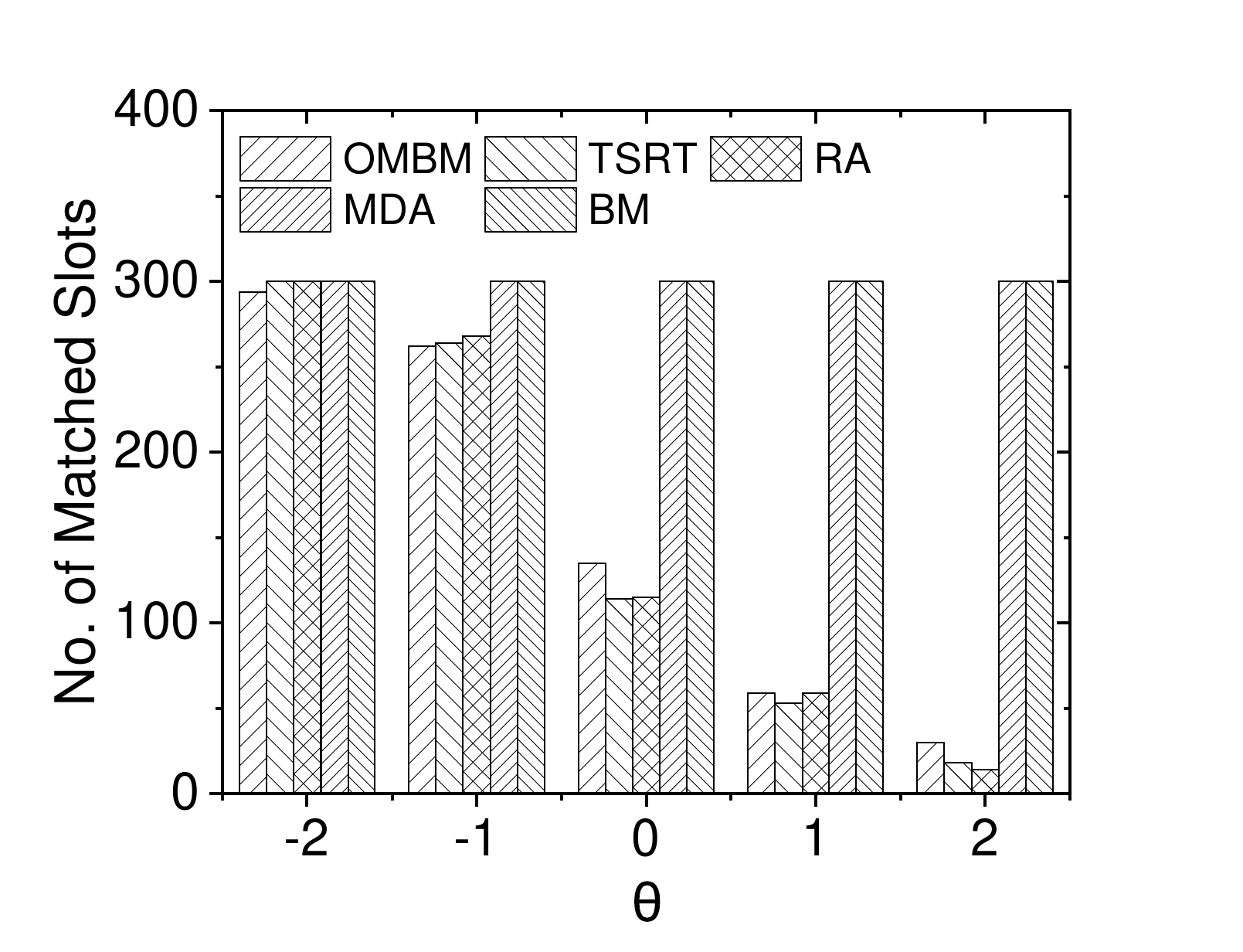} & \includegraphics[scale=0.17]{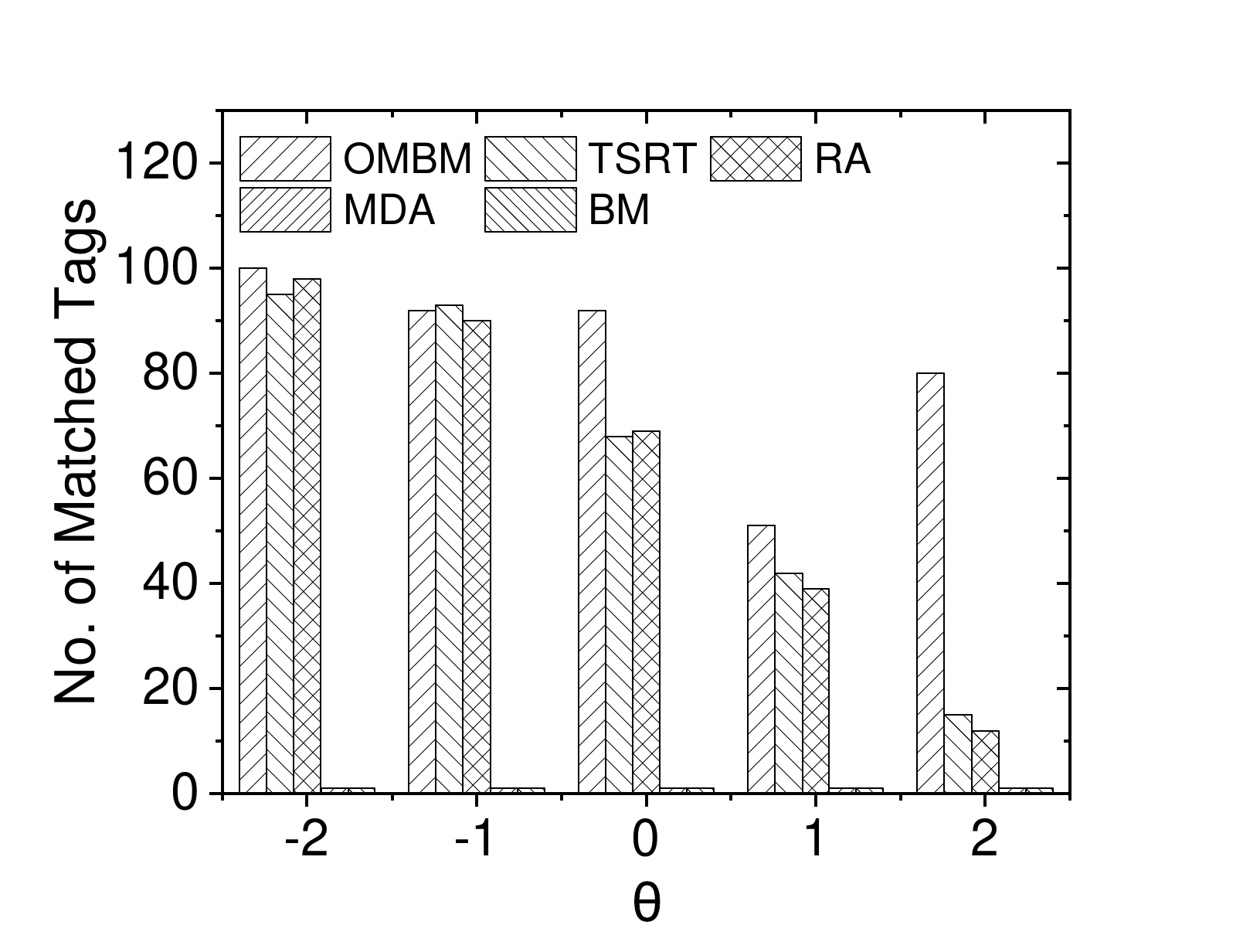} & \includegraphics[scale=0.17]{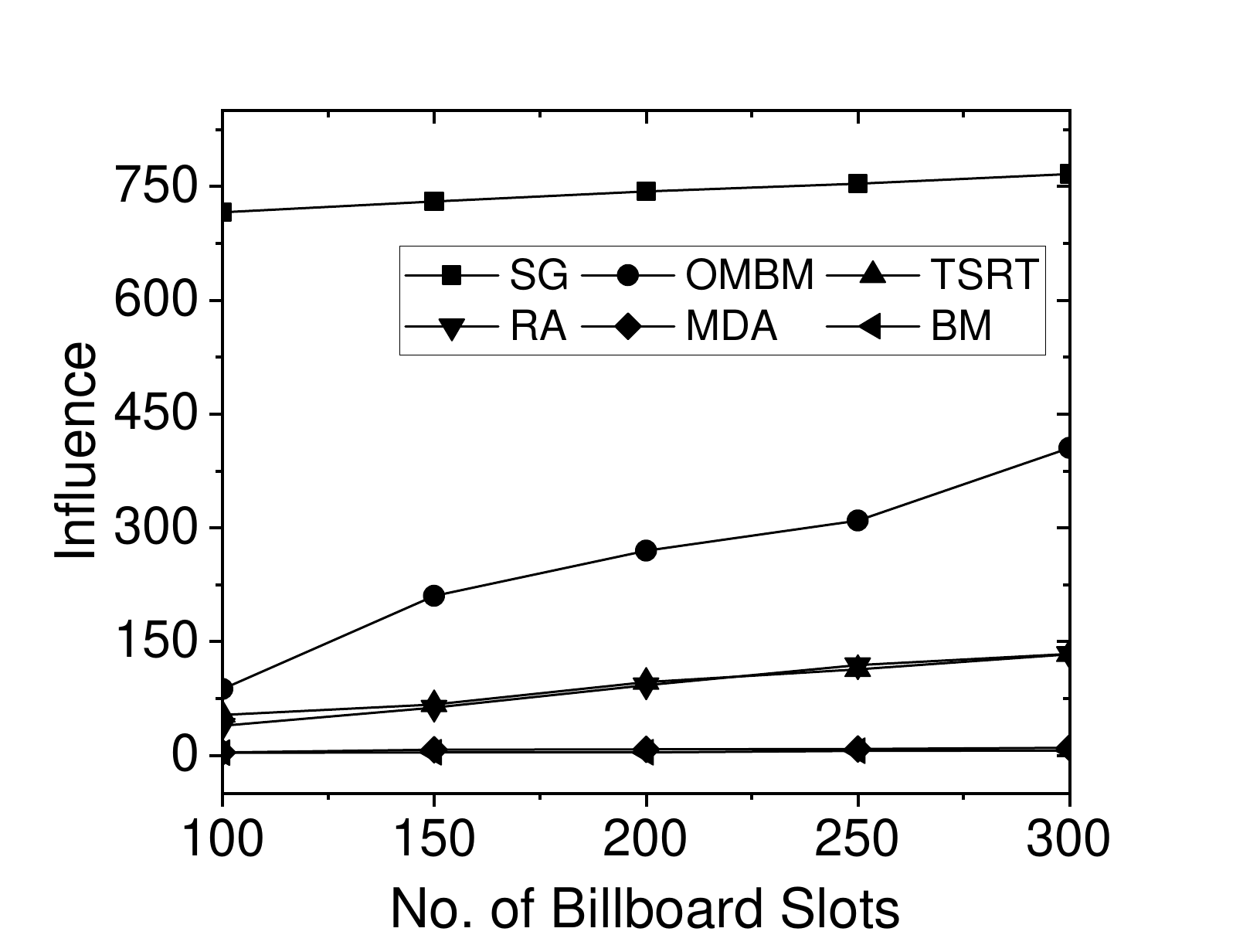} \\
(j)  Varying $\theta$ Vs. Slots (LA) & (k) Varying $\theta$ Vs. Tags (LA) & $(\ell)$ Slots Vs. Influence (LA) \\ 
\end{tabular}
\caption{ Varying billboard slots in NYC $(a,b,c,f)$, LA $(g,h,i, \ell)$, Varying $\theta$ in NYC $(d,e)$, LA $(j,k)$}
\label{Fig:Plot}
\end{figure*}

\paragraph{Tags, Slots Vs. Time.}
Figure \ref{Fig:Plot}(c) and \ref{Fig:Plot}(i) shows the efficiency of our proposed solution, and from this, we have three main observations. First, with the increase in slots, the run time for all the baseline and proposed approaches increases. The time requirement for the stochastic greedy approach is not presented in this paper; however, in NYC, it takes $152632$, $152725$, $153038$, $153046$, and $153125$ seconds for the number of slots $100$, $150$, $200$, $250$, and $300$, respectively. In LA run time is $10880$, $11039$, $11259$, $11283$, and $11289$ seconds for the same parameter setting. Second, the proposed `OMBM' takes more time than the baseline methods because it needs a comparison of both side vertices in the bipartite graph, i.e., the best match for each slot to tags and tags to slots, resulting in significant increases in run time. Third, the `RA' and `TSRT'  are the fastest methods because they need to compare each slot to tags only once.
\paragraph{Varying $\theta$ value.}  Figure \ref{Fig:Plot}(d,e) and Figure \ref{Fig:Plot}(j,k) show the impact of varying $\theta$ values on the NYC and LA datasets, respectively. We have three main observations. First, the parameter $\theta$ controls the pruning criteria in the bipartite graph in Algorithm \ref{Algo:TAP}, and accordingly, a dense or sparse bipartite graph is generated. Second,  the effectiveness and efficiency are very sensitive to varying $\theta$ values in the NYC and LA datasets for the `OMBM', `RA', and `TSRT' when we very billboard slots from $100$ to $300$. Third, the large $\theta$ value will increase the efficiency with worse effectiveness. This happens because when $\theta$ increases, the threshold for pruning increases, the graph becomes sparse, and the number of matched slots and tags becomes less.
\paragraph{Effectiveness Test.} Figure \ref{Fig:Plot}(f) and Figure \ref{Fig:Plot}($\ell$) shows the effectiveness of the proposed and baseline methods for  NYC and LA datasets, respectively. We have two main observations. First, stochastic greedy (SG) selects the required number of slots and tags, which returns maximum influence compared to the proposed and baseline methods. This occurs because the output of `SG' is the input of the proposed and baseline methods in the form of a bipartite graph, and after applying the $\theta$-score threshold, some slots and tags are pruned from that graph. Second, our proposed `OMBM' outperforms the baseline methods in allocating the most influential tags to appropriate slots, and among the baseline methods, `RA' and `TSRT' perform well compared to the other baselines.
\paragraph{Scalability Test.} 
To evaluate the scalability of our method, we vary the number of slots and tags from $100$ to $300$ and $25$ to $125$, respectively. As shown in Figure \ref{Fig:Plot}(c,i), the efficiency in `OMBM' is very sensitive compared to the baseline methods with a fixed number of tags and varying slot values. We observed that runtime also increases in both NYC and LA datasets with the increasing number of slots from $100$ to $300$ by $3\times$ to $5 \times$ more for the `OMBM' and baseline methods.
\paragraph{Additional Discussion.} 
Now, we discuss the effect of additional parameters in our experiments. First, we vary the number of slots $(k)$ between 100 and 300 and the number of tags $(\ell)$ between 25 and 125 to select the required slots and tags. In Figure \ref{Fig:Plot}, all results are presented for varying $k$ values while keeping the number of tags fixed. Second, in the stochastic greedy approach \cite{ali2024influential}, $\epsilon$ determines the size of the random sample subsets. This work uses $\epsilon = 0.01$ as the default setting. Also, we varied the $\epsilon$ value from 0.01 to 0.2 and observed that as $\epsilon$ increases, the run time of the stochastic greedy approach decreases; however, the quality of the solution degrades, i.e., the influence value decreases. Third, the influence of all the proposed and baseline methods increases when we vary the distance $(\lambda)$ from 25 meters to 150 meters. Both influence and run-time increase because one slot can influence more trajectories within its influence range. In our experiment, we set the value of $\lambda$ to 100 meters for both the NYC and LA datasets, as we observed that increasing $\lambda$ beyond 100 meters resulted in only marginal improvements. This paper does not report the effect of varying $\lambda$ and $\epsilon$ due to space limitations.
\section{Concluding Remarks}\label{Sec:Conclusion} 
In this paper, we have studied the Tag Assignment Problem in the context of Billboard Advertisement and proposed a one-to-many bipartite matching-based solution approach. The proposed methodology is illustrated with an example. An analysis of the proposed methodology has been done to understand its time and space requirements as well as its performance guarantee. Experiments with real-world datasets show that our approach outperforms baseline methods in terms of influence. This study can be extended by considering the following directions. Our study does not consider the `zonal influence constraint' which can be considered for future study. The methodology proposed in this paper works for a single advertiser. One important future direction will be to extend our study in multi-advertiser setting.


%
%
%
\bibliographystyle{splncs04}
\bibliography{Paper}

\end{document}